\documentclass[runningheads]{llncs}
\usepackage{a4wide} 
\usepackage{graphicx}
\usepackage{amsfonts}
\usepackage{amssymb}
\usepackage{amsmath}
\usepackage{xspace}

\newcommand{\N}{\ensuremath{\mathbb{N}}\xspace}
\newcommand{\Sset}{\ensuremath{\mathbb{S}}\xspace}
\newcommand{\Net}{\ensuremath{\mathcal{N}}\xspace}
\newcommand{\States}{\ensuremath{\mathit{States}}\xspace}
\newcommand{\PoW}{\ensuremath{\mathit{PoW}}\xspace}
\newcommand{\PoS}{\ensuremath{\mathit{PoS}}\xspace}
\newcommand{\ie}{\emph{i.e.}}

\DeclareMathOperator*{\argmax}{\mathbf{argmax}}

\newcommand\minority{\ensuremath{\frac{(|C|-1)}{2}}\xspace}

\begin{document}
\title{Stateless Distributed Ledgers\thanks{This research is partly supported by Japan Science and Technology Agency (JST) OPERA Grant Number JY280149.}}
%
%
\author{François Bonnet\inst{1}
\and
Quentin Bramas\inst{2}
\and
Xavier Défago\inst{1}
}
\authorrunning{F. Bonnet et al.}
%
\institute{
School of Computing, Tokyo Institute of Technology, Japan
\\\email{\{bonnet,defago\}@c.titech.ac.jp}
\and
ICUBE, University of Strasbourg, CNRS, France\\
\email{bramas@unistra.fr} {\small \it (corresponding author)}}
\maketitle              
\tolerance=400
\begin{abstract}
In public distributed ledger technologies (DLTs), such as Blockchains, nodes can join and leave the network at any time. A major challenge occurs when a new node joining the network wants to retrieve the current state of the ledger. Indeed, that node may receive conflicting information from honest and Byzantine nodes, making it difficult to identify the current state.

In this paper, we are interested in protocols that are \emph{stateless}, i.e., a new joining node should be able to retrieve the current state of the ledger just using a fixed amount of data that characterizes the ledger (such as the genesis block in Bitcoin).

We define three variants of stateless DLTs: weak, strong, and probabilistic. Then, we analyze this property for DLTs using different types of consensus.

\end{abstract}
\tolerance=200
\section{Introduction}

Distributed Ledger Technologies~(DLTs) are usually partitioned depending on the type of consensus used to order incoming transactions. Here, we consider the three most used classes of technologies: Byzantine agreement, Proof of Work, and Proof of Stake. 

Byzantine agreement protocols~\cite{raynal2013distributed} are used to maintain consistent states replicated over multiple servers. It can tolerate crash or Byzantine faults, up to a number that depends on the synchrony assumption of the communications. Such protocols are executed by known servers in a fixed network environment, a setting called \emph{permissioned}. They can easily be used by  nodes to maintain a ledger of transactions. Every insertion in the ledger is the result of a consensus among participating nodes.

Blockchains based on \emph{Proof of Work (PoW)} are the first distributed ledger technologies that work in an environment where nodes can join and leave and any node can participate to the protocol. Nodes are elected randomly proportionally to their computational power, and an elected node can append transactions to the ledger. All current PoW Blockchains, including Bitcoin, work with synchronous communications, and assume correct nodes to have strictly more computational power than Byzantine nodes.

Blockchains based on \emph{Proof of Stake (PoS)} are similar to the PoW based ones, but nodes are elected proportionally to the amount of tokens they own in the blockchain itself.
In protocols based on PoS, a well-known concern is called \emph{long-range} attacks~\cite{deirmentzoglou2019survey}, where a group of nodes create an alternative chain extending an old block. This is made possible because block generation is not computationally heavy, and if a node can extend a block at a given time, nothing prevents it from extending the same block in a different way at a later time. This attack becomes even worse when the nodes owning a majority of tokens at a previous time, do not have stake at the current time. Performing such attacks could be appealing as they have nothing to lose. This problem is known as posterior corruption~\cite{bentov2016snow}.

Existing Proof of Stake based protocols such as SnowWhite~\cite{bentov2016snow}, Algorand~\cite{gilad2017algorand}, Ouroboros~\cite{kiayias2017ouroboros}, have identified such risks. The main solution proposed is to have some sort of checkpointing mechanism to avoid considering past majorities of stake holders.
A variant of this attack is called stake-bleeding attacks~\cite{gavzi2018stake}. It uses other mechanisms such as block rewards and transaction fees to allow even a past minority of stake holders to execute long-range attacks.

\paragraph{Contributions.}
In this paper we aim at defining a general model to capture the main difference between various kind of Distributed Ledger Technologies (DLTs). Our model is abstract enough to be independent of the implementation and capture only the main mechanisms of the DLTs. Then, we focus on one property that we call the \emph{Stateless Property}. We define this property in our general DLT model and show whether existing technologies satisfies it or not. In particular the fact that Proof of Stake based DLTs are not Stateless implies the existence of the long-range attack vulnerability.

We believe that our model could be use independently to capture other properties and compare technologies in an abstract way.

\section{Model}
We consider that time is discrete and at each time $t\in\N$, $\Net_t$ represents the set of nodes in the network. We consider that communication is instantaneous and there is a communication link between any two nodes in the network at any given time. We also assume that each node is identified, and is able to securely sign messages.

A distributed ledger is a data structure with an \emph{``append''} function. It is maintained by a set of processing nodes. The network receives events and the nodes react to the events according to the distributed ledger protocol. Each time the ledger is updated, a new time instant begins. Formally, a DLT is characterized by its initial state $I$ and a state transition function $\sigma$ that takes a current state $S_t$, the events $E_t$, and the network $\Net_t$ containing all the nodes that are online at least once before the next ``append''. Then $\sigma$ returns a new state $S_{t+1}$ when the ``append'' function is called at time $t+1$. A state can be seen as a sequence of ``append'' and we write $S \preccurlyeq S'$ when state $S$ is a prefix of state $S'$, \ie, $S'$ can be obtained from $S$ by appending data. The state $S^{-k}$ denotes the truncated state $S$ where the last $k$ occurrences of ``append'' of $S$ are omitted. 

Given a DLT $(I, \sigma)$, a sequence of $\Net = (\Net_t)_{t\in\N}$ of networks and a sequence $E = (E_t)_{t\in\N}$ of events, we can construct the sequence of states $\States(I,\sigma, \Net, E) = (S_t)_{t\in \N}$ in the following way $S_0 = I$ and $\forall t\in\N, S_{t+1} = \sigma(S_{t}, \Net_t, E_t)$.

\paragraph{Stateless DLT.}
When a new node joins the network, it obtains the current state from the other nodes in the network. Informally, we say a DLT is stateless if a new joining node is able to deduce what is the current state of the DLT from the information received from the current network and by knowing the initial state.

At a given time $t$, each node $u\in\Net_t$ has a local state $LS(u)$. For a correct node, the local state is exactly the current state $S_t$ (communications are supposed instantaneous so all correct nodes agree on the current state at any time). For Byzantine nodes, the local state is constructed by an adversary. The set of pairs $(u,LS(u))$ for all nodes $u$ in $N_t$ is denoted $\Sset_t$ \ie, $\Sset_t = \{(u, LS(u)) \,|\,u\in\Net_t\}$.

\begin{definition}[Weakly Stateless DLT]
A DLT is weakly stateless if there exists a function $f$ such that $f(I, \Sset_t) = S_t$.
\end{definition}

\begin{definition}[Strongly Stateless DLT]
A DLT is strongly stateless if there exists a function $f$ such that $f(I, \Sset_t) = S_t$ and, for any subset $A\subset \Sset_t$, $f(I, A) = S_t$ or $\bot$.
\end{definition}

\begin{definition}[Probabilistically Stateless DLT]
A DLT is probabilistically stateless if there exists a function $f$ such that $\forall k,t,t'\in \N$, with $k\leq t \leq t'$, $f(I, \Sset_t)^{-k} \preccurlyeq S_{t'}$ with probability greater than $1-O(e^{-ck})$ for some constant $c>0$.
\end{definition}

\section{Examples of Stateless DLTs}

\subsubsection{Byzantine Agreement Protocols}
It is well-known that, in a fixed network $C$ of known nodes where communication is synchronous, consensus is possible and can tolerate up to \minority Byzantine nodes~\cite{raynal2013distributed}.

We denote by $\sigma_{BA}$ the transition function of a Byzantine agreement protocol among the nodes in $C$. $\sigma_{BA}$ represents the fact that, at time~$t$, the nodes in $C\subset \Net_t$ perform a Byzantine agreement to order the transactions received in $E_t$ and update the state~$S_t$ accordingly to obtain $S_{t+1}$.
In the Byzantine agreement protocol, we consider that the state contains the information about the set $C$ of nodes participating in the consensus protocol.

Interestingly, when a majority of nodes in $C$ are correct, any node $u$ outside $C$ can ask the nodes in $C$ for the current state. Then, the current state is the one received by a majority of nodes in $C$, and is guaranteed to be correct. From this we deduce the following theorem:

\begin{theorem}
For a set of nodes $C$ and an initial state $I$ (which contains the information of $C$), if $\forall t\in\N,\, C\subset \Net_t$ and at most \minority nodes in $C$ are Byzantine, then the DLT $(I,\sigma_{BA})$ is strongly stateless.
\end{theorem}
\begin{proof}
Assuming at most \minority nodes are Byzantine, for any sequence $\Net$ and $E$, all the correct nodes in $C$ agree on all the state $S_t \in \States(I, \sigma_{BA}, \Net, E)$ for any $t\in \N$.
The function $f$ returns the local state that appears in at least $\frac{|C|+1}{2}$ pairs associated with a node in $C$. If no such state exists (if the set of local states is only a subset of $\Sset$), $f$ returns $\bot$. Formally, we have\\
$f(I, A) = S$ if $|\{u \in C | (u,S)\in A\}| > \frac{|C|}{2}$, and $f(I,A) = \bot$ otherwise.
\qed\end{proof}

\subsubsection{Proof of Work Blockchains}

In PoW Blockchains, there are no assumptions about the nodes in the network, except that at any time $t$, in $\Net_t$, the computational power of the correct nodes is strictly greater than the computational power of the Byzantine nodes. Then, the transition function $\sigma_{\PoW}$ applies an ordering on the transactions received in $E_t$ (decided by some node in the network) and appends the block of transactions to state $S_t$ resulting in $S_{t+1}$. In addition, from state~$S_t$, one can compute the proof of work performed until time $t$, denoted $\PoW(S_t)$. So $\PoW(S_{t+1})$ is the sum of $\PoW(S_{t})$ and the proof of work corresponding to the last ``append''. For the initial state $I$, $\PoW(I) = 0$.

\begin{theorem}
For an initial state $I$, if for all $t$ the computational power of correct nodes in $\Net_t$ is strictly greater than the computational power of Byzantine nodes in $\Net_t$, then the DLT $(I,\sigma_{\PoW})$ is probabilistically stateless.
\end{theorem}

\begin{proof}
Let $f$ be the function returning the local state that maximizes the Proof of Work, $f(I, \Sset_t) = \argmax_S(\{\PoW(S) | \exists u, (u,S)\in\Sset_{t}\})$. Such a local state could have been generated by an adversary. If $k$ denotes the number of blocks we have to truncate to obtain a prefix of the correct state $S_t$, then the probability decreases exponentially fast when $k$ tends to infinity. Indeed, let $p_t$, resp. $q_t$, be the computational power of the correct nodes, resp. of the adversary, at time $t$. By assumption, $\forall t, p_t>q_t$. Let $\lambda_t = \max_{t'\leq t}(p_{t'}q_{t'})$.

From~\cite{grunspan2018double} (Th.2), we deduce that, at a given time $t$, the probability that an adversary rewrites the last $k$ blocks is in $O(e^{-cz})$ with $c=\log(1/(4\lambda_t)) > 0$.
\qed\end{proof}

\section{Impossibility of Stateless Proof of Stake Blockchains}

In PoS Blockchains, the consensus at a given time $t$ is possible assuming the nodes owning a majority of the tokens are correct. So we can assume that the state transition function $\sigma_{\PoS}$ is performed by those correct nodes, owning a majority of the tokens, creating a new state $S_{t+1}$ from state $S_t$ and incoming events $E_t$. However nothing prevents the nodes in $\Net_0$ to create an alternative state after time $0$.

\begin{theorem}
Even if, at each time, all the nodes owning tokens are correct, the PoS DLT is not weakly stateless.
\end{theorem}
\begin{proof}
To simplify, assume that, at some time $t> 0$, all nodes in $\Net_0$ are owned by the adversary. Indeed, there is no assumption on the correctness of $\Net_0$ after time $0$. Then, the adversary can simulate an execution of the DLT in a sequence of networks $\Net_t'$ where each node in $\Net_t'$ is owned by the adversary and such that there is a bijection $m$ mapping any nodes in $u \in \Net_t$ to a malicious node $m(u)\in \Net_t'$, for all $t\in\N$ ($m(u) = u$ if $u$ is already malicious). Hence, the sequences of networks $(\Net_t')_t$ and $(\Net_t)_t$ differ only in the addresses that identify nodes. The adversary can execute the same DLT using the same sequence of events $(E_t)_t$ but in the malicious sequence of networks $\Net'=(\Net_t')_t$, which gives a different sequence of states $\States(I,\sigma_{\PoS}, \Net', E) \neq \States(I,\sigma_{\PoS}, \Net, E)$. 

At time $t$, the adversary can connect all nodes $\Net_t'$ to the network so that the set of local states is $\Sset_t = \{(u, LS(u)) | u\in \Net_t\} \cup \{(u,LS(u)) | u\in \Net_t'\}$. The set is symmetric as half of the local states contain $S_t$ and the other half contain $S_t'$ and both states differ only in the addresses used to identify nodes.

$\States(I,\sigma_{\PoS}, \Net', E)$ is a valid sequence of states if the sequence of networks was $\Net'$ and if the adversary creates symmetrically the state $S_t$ using the sequence of networks $\Net_t$. A function $f$ should answer $S_t$ in the first case and $S_t'$ in the second case, with exactly the same input, which is a contradiction.
\qed\end{proof}

One way to make the PoS DLT stateless is to assume that any set of nodes owning a majority of the token at a given time $t$ are correct at any time $t'> t$. This is a very strong assumption as for instance, it gives the same kind of trust in the initial set of nodes as in the Byzantine agreement protocol. Indeed, if the nodes in $\Net_0$ are still connected at time $t$ then, they act as a trusted committee. We could even remove the proof of stake entirely and only rely on standard Byzantine agreement between nodes in $\Net_0$. In the case where nodes in $\Net_0$ go offline, they can select replacement nodes (using any method including election or simply one-to-one replacement), and include a signed message that define their choice in the state so that any joining node could identify the current set of trusted nodes. In future work, we plan to identify other sufficient conditions for the existence of Stateless DLTs.

\bibliographystyle{plain}
\bibliography{biblio}

\end{document}